\documentclass[envcountsame,envcountsect,nolinenumbers]{lnthi}

\Tsetundated

\usepackage{graphics,hyperref}

\def\Tvn#1{\hbox{\textit{#1\/}}}

\def\Tceil#1{\lceil #1\rceil}

\title{On-the-Fly Array Initialization in Less Space}

\author{Torben Hagerup\inst{1} \and Frank Kammer\inst{2}}
\institute{Institut f\"ur Informatik, Universit\"at Augsburg, 86135
Augsburg, Germany\\\email{hagerup@informatik.uni-augsburg.de}
\and MNI, Technische Hochschule Mittelhessen,
 35390 Gie\ss en, Germany\\\email{frank.kammer@mni.thm.de}}

\begin{document}
\overfullrule=5pt
\maketitle{}

\pagestyle{plain}
\thispagestyle{plain}

\begin{abstract}
We show that for all given $n,t,w\in\{1,2,\ldots\}$
with $n<2^w$,
an array of $n$ entries of $w$ bits each
can be represented on a word RAM with a
word length of $w$ bits in at most
$n w+\lceil{n({t/{(2 w)}})^t}\rceil$ bits of uninitialized
memory to support constant-time initialization
of the whole array
and $O(t)$-time reading and writing of
individual array entries.
At one end of this tradeoff, we achieve
initialization and
access (i.e., reading and writing)
in constant time
with $n w+\lceil{{n/{w^t}}}\rceil$ bits for arbitrary fixed~$t$,
to be compared with $n w+\Theta(n)$ bits for the
best previous solution,
and at the opposite end,
still with constant-time initialization,
we support $O(\log n)$-time
access with just $n w+1$ bits, 
which is optimal for arbitrary access times
if the initialization executes fewer than
$n$ steps.\\

\end{abstract}

\section{Introduction}
\label{sec:intro}
Whereas the space used by an algorithm
(measured in ``memory units'' such as
words) is usually
bounded by its running time, there may be exceptions
if the memory offers random access, and it is
occasionally useful to employ large arrays of which
only a small part will ever be accessed.
A case in point are adjacency matrices, which
are a convenient representation of graphs if the
algorithms to be executed issue
adjacency queries (e.g., ``does $G$ contain an edge
from $u$ to $v$?'')\
in an irregular pattern that cannot
be served efficiently using adjacency lists.
Even if one can afford the space needed by an
adjacency matrix, it may be prohibitively expensive
to clear all those entries in the matrix that do not
correspond to edges in the graph.
The problem does not occur if the memory cells allocated to
hold the adjacency matrix can be assumed
to be already initialized to some particular value
(that can be taken to signify ``no edge''),
but in general this is not a realistic assumption.
Therefore the problem of simulating an initialized
array in an uninitialized memory has been considered
since the early days of computing.

Additional motivation for our work comes from
the fact that certain modern programming languages
such as Java, VHDL and~D stipulate that memory be
initialized (e.g., cleared to zero)
before it is allocated to application programs~\cite{GosJSBB15,VHDL11}
or have this as the default
behavior~\cite{Ale10}.
The initialization is carried out
for security reasons and to ease debugging
by making faulty programs more
deterministic.
If it can be ensured that application programs
access memory only through a well-defined interface,
one may hope to let the interface provide conceptually cleared memory
while avoiding the overhead of clearing
the memory physically.

For some $w\in\mathbb{N}=\{1,2,\ldots\}$, our model
of computation is a word RAM~\cite{AngV79,Hag98} with a word length
of $w$ bits,
where we assume that $w$ is
large enough to allow all memory words in use to be addressed.
As part of ensuring this, in the context of
an array of size $n$ we always assume that
$n<2^w$.
The word RAM has constant-time operations
for addition, subtraction and multiplication
modulo $2^w$, division with truncation
($(x,y)\mapsto\lfloor{{x/y}}\rfloor$ for $y>0$),
left shift modulo $2^w$
($(x,y)\mapsto (x\ll y)\bmod 2^w$,
where $x\ll y=x\cdot 2^y$),
right shift
($(x,y)\mapsto x\gg y=\lfloor{{x/{2^y}}}\rfloor$),
and bitwise Boolean operations
($\textsc{and}$, $\textsc{or}$ and $\textsc{xor}$
(exclusive or)).
We also assume a constant-time operation to
load an integer that deviates from $\sqrt{w}$
by at most a constant factor---this enables the
proof of Lemma~\ref{lem:word}.
The problem of central concern to us is to
realize a clearable word array,
defined as follows:

\begin{definition}
A \emph{clearable word array} is a data structure
that can be initialized with an integer
$n\in\mathbb{N}$ and subsequently maintains an
element of $\{0,\ldots,2^w-1\}^n$, called its
\emph{client sequence} and initially
$(0,0,\ldots,0)$, under the following operations:

\begin{description}
\item[\normalfont$\Tvn{read}(\ell)$]
($\ell\in\{0,\ldots,n-1\}$):
If the client sequence before the call is
$(x_0,\ldots,x_{n-1})$, returns $x_\ell$
without changing the client sequence.
\item[\normalfont$\Tvn{write}(\ell,x)$]
($\ell\in\{0,\ldots,n-1\}$ and $x\in\{0,\ldots,2^w-1\}$):
If the client sequence before the call is
$(x_0,\ldots,x_{n-1})$, changes the client
sequence to be
$(x_0,\ldots,x_{\ell-1},x,x_{\ell+1},\ldots,x_{n-1})$.
\end{description}
\end{definition}

The clearable word array is a special case
of the \emph{initializable array} of
Navarro~\cite{Nav14}.
There are two differences.
First, the data structure of Navarro is more
general in that the initialization, in addition
to~$n$, receives a second parameter $v$ that
is taken to be the initial value of the
array entries, i.e., the initial value of the
client sequence is $(v,v,\ldots,v)$ rather
than $(0,0,\ldots,0)$.
As is easy to see and will be discussed in
Section~\ref{sec:contribution},
however, the more general data structure
reduces easily to the more restricted one.
Second, Navarro does not specify the nature of
the array entries, which is of no relevance
to his approach, whereas we fix the array entries to be
\emph{words}, i.e., elements of $\{0,\ldots,2^w-1\}$.
Again, this will turn out to be a
restriction of little consequence.

Following the initialization of a clearable
word array with an integer $n$, we call $n$
the \emph{universe size} of the data structure.
We shall have occasion to consider restricted
clearable word arrays that can be initialized only
for certain specific universe sizes.
Because the connection between the client sequence
of an initializable array and an array used to hold it
is often very close, it is easy to confuse the two.
We may view the client sequence as an array, but
then use the letter `$a$' to denote this abstract array
(which is initialized)
and `$A$' to denote the corresponding physical array
(which is not initialized).

\section{Previous Work}
\label{sec:related}
Fredriksson and Kilpel\"ainen~\cite{FreK16} give a
detailed overview of the known approaches to
array initialization
and compare them experimentally.
In the discussion of their work, we assume that
the task is to realize an initializable array
of $n$ entries of $b\le w$ bits each.
Define the \emph{redundancy} of a data structure
that solves this problem and occupies $N$ bits
to be $N-n b$, i.e., the number of bits
used beyond the minimum of $n b$ bits
needed even without the requirement
of initializability.

A number of the methods described by
Fredriksson and Kilpel\"ainen can be viewed as
special cases of a general \emph{trie method}.
Ignoring rounding issues, the trie method is
parameterized by an integer $h\in\mathbb{N}$
and a \emph{degree sequence} $(d_1,d_2,\ldots,d_h)$
of $h$ positive integers with $\prod_{i=1}^h d_i=n b$.
It uses a tree $T$ of height~$h$ in which all
nodes of height $i$ have $d_i$ children,
for $i=1,\ldots,h$.
Each node in $T$ has an associated bit,
the bits of each maximal group of siblings are stored
compactly, $w$ bits to a word, and
the $n b$ bits at the leaves are identified with
the $n b$ bits of the abstract array~$a$.

Let \emph{processing} an inner node $u$ in $T$ be the
following:
If the bit associated with $u$ has the value~0
(informally, $u$ has been initialized, but
its children have not),
initialize the bits of all children of~$u$,
to 0 if the children are inner nodes and
to the prescribed initial value $v$---within
groups of $b$ siblings in the obvious manner---if
they are leaves.
If $u$ has $d$ children, this can be done in
$O(\lceil{{d/w}}\rceil)$ time.
Finally set the bit associated with $u$ to~1.
If the value of that bit is 1 already prior to the
processing of~$u$, the processing of~$u$
terminates immediately after discovering this fact.

To initialize $T$, set the bit at its root to~0.
In addition, it is permissible, as
part of the initialization, to process the
inner nodes in an upper part of $T$ in a top-down fashion,
i.e., so that no nonroot node is processed
before its parent.
We will say that such nodes are
\emph{preprocessed}.
To read the $\ell$th entry of~$a$, descend in $T$
towards the $\ell$th group of $b$ leaves.
If an inner node is encountered whose
associated bit has the value~0,
return $v$.
If not, return the value found in the
$\ell$th group of $b$ leaves.
To write the $\ell$th entry of~$a$, 
descend in the same manner towards the
$\ell$th group of $b$ leaves, process every inner node
encountered on the way, and finally store the
appropriate value in the bits of the $\ell$th
group of $b$ leaves.
The total number of bits used by the data structure
is the number of nodes in $T$ that are not preprocessed,
the initialization takes constant time plus
time proportional to the sum of $\lceil{{d/w}}\rceil$
over all degrees $d$
of preprocessed nodes,
the worst-case time of \Tvn{read} is
$\Theta(h)$, and the worst-case time of \Tvn{write} is
the maximum over all leaves $v$ in $T$ of
$\Theta(h+\sum_i\lceil{{d_i}/w}\rceil)$,
where the sum ranges over those values of~$i\in\{1,\ldots,h\}$
for which the ancestor of $v$ of height~$i$
is not preprocessed.

Fredriksson and Kilpel\"ainen consider the following
special cases of the trie method:
Degree sequence $(n b)$, preprocess the root
(Plain);
degree sequence $(b,n)$, preprocess the root
(Simple);
degree sequence $(b,w,w,\ldots,w)$ (Hierarchic);
degree sequence $(b,{n/w},w)$ (Simple-H); and
degree sequence $(b,w,{n/w})$, preprocess
the root (SHV).
The redundancy is 0
for Plain
and close
to $n$ (i.e., the number of nodes
in $T$ of height~1) for the other methods.
The initialization time is $\Theta(1+{{n b}/w})$
for Plain, $\Theta(1+{n/w})$ for Simple,
$\Theta(1+{n/{w^2}})$ for SHV and
$\Theta(1)$ for the other methods.
The worst-case time for \Tvn{read} is $\Theta(1+\log_w n)$
for Hierarchic and $\Theta(1)$ for the other methods.
The worst-case time for \Tvn{write}, finally, is
$\Theta(1+\log_w n)$ for Hierarchic,
$\Theta(1+{n/{w^2}})$ for Simple-H and
$\Theta(1)$ for the other methods.

None of the methods discussed above combines
constant initialization time with constant access time,
and it is easy to see that this is true of every
instance of the trie method.
Constant time for every operation is achieved by
a folklore method
that goes back at least to the
early 1970s (see~\cite[Exercise~2.12]{AhoHU74}).
The folklore method uses a physical
array $A$ with the index set $\{0,\ldots,n-1\}$
and assigns the codes $0,1,\ldots$
to the indices of the abstract array $a$ in the order in which the
indices are first used in calls of \Tvn{write},
$x_\ell$ is stored in $A[f(\ell)]$, where $f(\ell)$ is
the code of $\ell$, two tables are used
to keep track of the encoding function $f$
and its inverse $f^{-1}$,
and finally the data structure remembers the
number $k$ of codes assigned.
To access $x_\ell$, first $f(\ell)$
is looked up in the table of~$f$.
Because the table is not initialized, the purported
code $j$ may not be correct, but $j$ is the
code of $\ell$ exactly if $0\le j<k$ and the entry
of $j$ in the table of $f^{-1}$ is $\ell$.
If not, the default initial value $v$ is returned in the case of
a read operation, and
the next available code is assigned to~$\ell$
in the case of a write operation.
The remainder of the access is simply a reading
or writing of $A[f(\ell)]$.
The structure is initialized by setting $k$ to~0.
In addition to the space needed to hold the actual
data in $A$, it needs space for the tables of
$f$ and $f^{-1}$ and
the counter $k$,
so that its redundancy is
$2 n\lceil{\log_2 n}\rceil+\lceil{\log_2(n+1)}\rceil$.

A family of methods due to Navarro \cite{Nav14}
combines the Hierarchic method above with the
folklore method.
The idea is, starting from Hierarchic, to replace the
nodes of height $\ge h+2$, for some $h\ge 0$, by an
instance of the folklore data structure.
This achieves the same effect as processing the
nodes that were removed
and obviates the need to
descend through these nodes during an access.
The initialization time is constant, the worst-case
access time is $\Theta(h+1)$, and the
redundancy is approximately $3 n$ for $h=0$ and
approximately $n$ for $h\ge 1$.

\section{Our Contribution}
\label{sec:contribution}

We give an upper-bound tradeoff that spans the entire
range from minimal time to minimal space.
Our main result is the following:

\begin{theorem}
\label{thm:main}
There is a clearable word array
that, for all given $n,t\in\mathbb{N}$,
can be initialized for universe size $n$ in constant time
and subsequently occupies at most
$n w+\lceil{n({t/{(2 w)}})^t}\rceil$
bits and supports $\Tvn{read}$ and $\Tvn{write}$
in $O(t)$ time.
\end{theorem}

If $w$ and hence (by assumption) $n$ are
bounded by constants, it is trivial to
realize a clearable word array
with constant initialization
and access times and zero redundancy
(initialize the array explicitly, i.e.,
use the Plain method of
Fredriksson and Kilpel\"ainen).
Given a constant $t\in\mathbb{N}$, we can therefore
assume without loss of generality that
$w\ge t^2$.
Then $({t/w})^2\le{1/w}$
and hence $({{2 t}/{(2 w)}})^{2 t}\le{1/{w^t}}$.
Theorem~\ref{thm:main}
(used with $t$ doubled) thus implies that for
all constant $t\in\mathbb{N}$, there is a clearable
word array that can be initialized in constant
time, executes accesses in constant time
and has redundancy $\lceil{{n/{w^t}}}\rceil$.
The best previous constant-time solution, due to
Navarro~\cite{Nav14} and discussed above,
has \mbox{redundancy $n+o(n)$.}

At the other end of the time-space tradeoff,
for $t=\lceil{\log_2 n}\rceil$, the
redundancy of
Theorem~\ref{thm:main} is $1$, i.e., the
constant-time initialization costs only a single bit
and accesses are still supported in
logarithmic time.
If an initialization time of $\Theta(n)$ is
acceptable, a clearable word array with
constant-time access can obviously
be realized with zero redundancy---this is
again the Plain method of Fredriksson
and Kilpel\"ainen.
On the other hand,
the redundancy cannot be reduced below our
bound of~1 for any access times
unless the initialization writes to at least $n$ words,
which needs at least $n$ steps.
To see this, assume that a clearable word array
with universe size $n$ is represented in $N$ bits for
some $N\in\mathbb{N}$.
Because the client sequence can be in any one
of $2^{n w}$ states, any two of which can be
distinguished through \Tvn{read} operations,
whereas its representation can be in
only $2^N$ states, we must have $N\ge n w$,
irrespectively of all operation times.
Moreover, if $N=n w$, every state of the client
sequence is represented by exactly one
bit pattern of its representation.
Since the client sequence is in a well-defined
state immediately after the initialization,
this is impossible unless each of the $n w$ bits
of its representation is forced to one specific
value during the initialization, i.e., unless the
initialization writes to at least $n$~words.

Note that it is a responsibility of the user of
a clearable word array initialized for
universe size $n$ to ensure that $\ell<n$
in all calls of the form $\Tvn{read}(\ell)$
or $\Tvn{write}(\ell,x)$ issued
to the data structure.
Whereas the data structure can easily 
check the conditions $\ell\ge 0$ and $0\le x<2^w$,
when operated close to its minimum space it cannot
afford to store the integer~$n$.
Thus illegal calls of its operations may go
undetected and may lead to attempted accesses
to memory words outside of the area
assigned to the data structure.

Our result can be seen as a second application of the
\emph{light-path technique}, which was introduced
(but not named) in~\cite{HagK16}
and used there to construct
space-efficient nonsystematic
choice dictionaries.
From a technical perspective, the situation is simpler here,
as there is no need to store data in a particular
\emph{compact representation} and to provide
conversion to and from the compact representation.
This gives us an opportunity to illustrate the
light-path technique in a purer setting.
At a more abstract level, the fundamental idea is to
upset the structure of a simple table slightly in order to
accommodate additional information in the table.
Whereas this principle has been used before
\cite{FiaMNS91,FraG08,Mun86},
curiously, it has not so far been employed in the
setting of initializable arrays even though it seems
particularly natural there.
It may be noted that the
$c$-color choice dictionaries
of \cite{HagK16} could be
used directly as initializable arrays,
but efficiently so only for arrays whose elements
are drawn from a very small range $\{0,\ldots,2^b-1\}$.
This is because each element of that range
would be considered a separate color, i.e.,
we would have $c=2^b$.

Given the clearable word array of
Theorem~\ref{thm:main}, it is easy to derive a more
general data structure that, for some integer $b$
with $1\le b\le w$, maintains a
client sequence in $\{0,\ldots,2^b-1\}^n$,
initially $(0,0,\ldots,0)$, under reading
and writing of individual elements of
the sequence.
Simply pack the $n$ elements of the client sequence
tightly in $\lceil{{{n b}/w}}\rceil$ words of $w$ bits each,
initialize the used part of the last word to 0,
maintain the other words in a clearable word array,
inspect a $b$-bit element of the client sequence
by reading the at most two words over which the
$b$ bits spread, picking out the relevant pieces of the words
and concatenating the pieces, and update a $b$-bit
element of the client sequence correspondingly
by splitting the new value into at most two pieces
and storing each piece appropriately in a word
without disturbing the rest of the word.
The execution times are within a constant factor of
those of the clearable word array, and the number of bits needed is
at most $n b+\lceil{n({t/{(2 w)}})^t}\rceil$.

We can also easily derive a data structure more
general than that of Theorem~\ref{thm:main} in that
the client sequence is initialized to
$(g(0),\ldots,g(n-1))$,
where $g:\{0,\ldots,n-1\}\to\{0,\ldots,2^w-1\}$
is some function, rather than to $(0,0,\ldots,0)$.
The simple idea is to swap the representations
of the ``internal'' and ``external'' initial values.
Both $\Tvn{read}(\ell)$ and $\Tvn{write}(\ell,x)$
then begin by evaluating $g(\ell)$.
If reading the value associated with $\ell$ in a
normal clearable word array yields the value~0,
$\Tvn{read}(\ell)$ returns $g(\ell)$.
If the value read is $g(\ell)$, $\Tvn{read}(\ell)$
returns~0, and every other value read is
returned as it is.
Similarly, if $x=g(\ell)$, $\Tvn{write}(\ell,x)$
actually writes the value~0 to the normal
clearable word array, $x=0$ causes the value
$g(\ell)$ to be written, and every other
value of $x$ is written as it is.
The initialization and access times are those of
Theorem~\ref{thm:main} plus whatever time is needed to
initialize $g$ and to evaluate it on one argument,
respectively, and the space requirements are those of
Theorem~\ref{thm:main} plus those of $g$.
It is easy to see that the generalizations described
in this and the previous paragraph can be
combined.

Very recently, giving a clever twist to the
folklore method, Katoh and Goto~\cite{KatG17}
devised a clearable word array that executes
every operation in constant time but,
when the universe size is~$n$,
uses just $n w+1$ bits.
Yet another solution was indicated by
Loong, Nelson and Yu~\cite{LooNY17}.

\section{The Construction}

In this section we prove Theorem~\ref{thm:main}.
At a very low and technical level, we need the following
staple of word-RAM computing.

\begin{lemma}[\cite{FreW93,HagK16}]
\label{lem:word}
Given a nonzero integer $\sum_{i=0}^{w-1}2^i b_i$,
where $b_i\in\{0,1\}$ for $i=0,\ldots,w-1$,
constant time suffices to compute $\max I$
and $\min I$, where $I=\{i\mid 0\le i\le w-1$ and
$b_i=1\}$.
\end{lemma}

Let a \emph{colored tree} be an ordered outtree, each of
whose leaves is either \emph{white} or \emph{black}.
Given a colored tree $T$, we extend the colors at
the leaves of $T$ to its inner nodes as follows:
If the leaf descendants of an inner node $u$
all have the same color (white or black),
then $u$ has that same color.
If $u$ has both a white and a black leaf
descendant, $u$ is \emph{gray}.
Clearly every ancestor of a node $v$ has the
same color as $v$ or is gray.
In particular, every ancestor of a gray node is gray.
Define the \emph{navigation vector} of an inner node
to be the sequence of the colors of its children
in the order from left to right.

\begin{figure}[b!]
\begin{center}
\scalebox{0.96}{%
\includegraphics{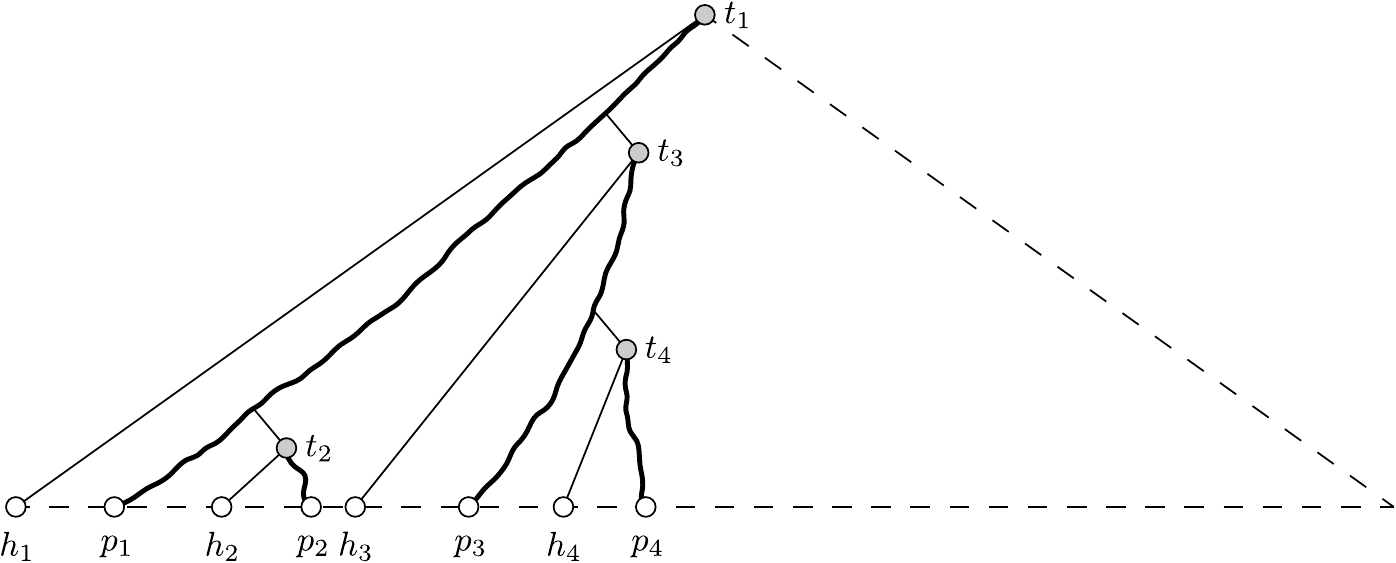}}
\end{center}
\caption{Example light paths (drawn thicker).
Top nodes, historians and proxies are labeled
``$t$'', ``$h$'' and ``$p$'', respectively,
and a subscript identifies the associated
light path.}
\label{fig:lightpaths}
\end{figure}

Recall that the \emph{left spine} of a rooted
ordered tree $T$ is the maximal path in $T$
that starts at the root of $T$ and,
whenever it contains an inner node $u$,
also contains the leftmost child of~$u$.
Define the \emph{preferred child} of a white or gray
inner node in a colored tree $T$ to be its leftmost
gray child if it has at least one gray child,
and its leftmost white child otherwise.
Call an edge in $T$ \emph{light} if it leads
from a gray inner node to its preferred child
or lies on the left spine of a subtree
of $T$ whose root is white and has a gray parent of which
it is the preferred child.
In other words, every gray inner node picks
the edge to its preferred child to be light,
whereas a white inner node does so only if
``prompted'' by its parent.
The light edges induce a collection of node-disjoint
paths called \emph{light paths}, each of which
ends at a leaf in~$T$.
When $P$ is a light path that starts at a
(gray) node $u$ and ends at a (white)
leaf $v$, we call $u$ the \emph{top node},
$v$ the \emph{proxy} and the leftmost leaf
descendant of $u$ (that may coincide with $v$)
the \emph{historian} of $P$ and of every node on $P$.
These concepts are illustrated in Fig.~\ref{fig:lightpaths}.
A gray node that is not the root of~$T$
is a top node exactly if it is not the preferred
child of its parent, i.e., if it has at least
one gray left sibling.
No proper ancestor of a top node $u$ can have
a descendant of $u$ as its leftmost leaf descendant,
so a leaf is the historian of at most one light path.
If $h$ is the historian of a light path $P$, the
top node and the proxy of $P$ are also said to be the
top node and the proxy, respectively, of $h$.
A leaf $\ell$ cannot be the historian
of one light path and the proxy of another,
since otherwise the two corresponding
top nodes would both be ancestors of $\ell$
and the path between them would contain
only gray nodes and be
part of a light path, a contradiction.
A similar argument shows that in the
left-to-right order of the leaves of~$T$,
no historian or proxy lies 
strictly between a historian and its proxy.
Define the \emph{history} of a light path that
contains the nodes $u_1,\ldots,u_k$, in that
order, to be the sequence $(q_1,\ldots,q_{k-1})$,
where $q_i$ is the navigation vector of $u_i$,
for $i=1,\ldots,k-1$ ($u_k$, as a leaf,
has no navigation vector).

The following lemma describes the
work-horse of our data structure.

\begin{lemma}
\label{lem:tree}
Let $d$ and $t$ be given positive integers with
$2 d t\le w$ such that $d$ is a power of\/~$2$.
Then there is a clearable word array that
can be initialized for universe size $n=d^t$
in constant time and subsequently occupies
$n w+2$ bits and, if given access
to the parameters $d$ and $t$, supports
\Tvn{read} and \Tvn{write} in $O(t)$ time.
\end{lemma}

\begin{proof}
Without loss of generality assume that $d\ge 2$.
We use a conceptual colored tree $T$ that is
a complete $d$-ary tree of height $t$ and identify
the $n$ leaves of $T$, in the order from left to right,
with the integers $0,\ldots,n-1$.
Let $r$ be the root of $T$ and,
for each node $u$ in $T$, let $T_u$ be the
maximal subtree of $T$ rooted at~$u$.
We represent a node $u$ of height $j$ in $T$
through the pair $(j,k)$, where $k$ is the number
of nodes in $T$ of the same height as $u$
and strictly to its left
(in other words, the nodes on each level in $T$
are numbered consecutively in the order from left
to right, starting at~0).
Then navigating in $T$ is easy:
If $u$ is not the root $r$, its parent is
(represented through) $(j+1,\lfloor{k/d}\rfloor)$,
if $u$ is not a leaf, its children are
$(j-1,k d),\ldots,(j-1,k d+(d-1))$,
$u$'s leftmost leaf descendant is $(0,k d^j)$
(identified with the integer $k d^j$),
and if $u$ is not a leaf and $\ell$ is a
leaf descendant of $u$, then $\Tvn{viachild}(u,\ell)$,
the child of $u$ that is an ancestor of $\ell$,
is $(j-1,\lfloor{{\ell/{d^{j-1}}}}\rfloor)$.
The assumption that $d$ is a power of~$2$ ensures
that we can compute the necessary powers of $d$
in constant time by means of multiplication
and left shift.
This requires the availability of $\log_2 d$,
which can be computed from $d$ in constant time
according to Lemma~\ref{lem:word}.

The actual data is stored in a word array $A$ with
index set $\{0,\ldots,n-1\}$ and in two additional
\emph{root bits}.
The three colors white, gray and black
are encoded in two bits, the navigation vector of an inner
node in $T$ is represented by the $2 d$-bit concatenation of 
the representations of its $d$ (color) elements, and the
history of a $k$-node light path is represented by the
$2 d (k-1)$-bit concatenation of the representations of
its $k-1$ (navigation-vector) elements.
The relation $2 d t\le w$ ensures that every history
fits in a $w$-bit word.
Assume that a history of
fewer than $w$ bits is ``right-justified'' in the
word so that the position in the word of the
navigation vector of a node depends only
on the height of the node.

With the aid of an algorithm of Lemma~\ref{lem:word},
the preferred child
of a given white or gray inner node $u$ in $T$ can be computed
in constant time from the navigation vector of $u$ or
a history that contains that navigation vector.
This may need a couple of bit masks (informally, ones
that correspond to all nodes having the same color)
that can easily be obtained via multiplication
with the integer
$1_{d t,2}={{(2^{2 d t}-1)}/3}$, whose $(2 d t)$-bit
binary representation is $0101\cdots 0101$.
Because $2^{2 d t}$ may not be representable in a $w$-bit word
(namely if $2 d t=w$),
the computation of $1_{d t,2}$ needs a little care,
but is still easy to do in constant time.

The client sequence $(x_0,\ldots,x_{n-1})$ is represented
in $A[0],\ldots,A[n-1]$ and the two root bits
according to the following storage invariants:
First, the two root bits indicate the color of the root $r$ of~$T$.
Second, for $\ell=0,\ldots,n-1$,

\begin{itemize}
\item
if $\ell$ is a historian, $A[\ell]$ stores the history of the proxy of $\ell$
(hence the term ``historian''),
\item
if $\ell$ is black and not a historian,
$A[\ell]$ stores $x_\ell$,
\item
if $\ell$ is a proxy whose historian $h$ is black,
$A[\ell]$ stores $x_h$
(as a ``proxy'' for $h$), and
\item
if $\ell$ is white and neither a historian nor a proxy
whose historian is black, the value of $A[\ell]$
may be arbitrary.
\end{itemize}

Note that because every proxy is white, for each $\ell\in\{0,\ldots,n-1\}$
exactly one of the four cases above applies.
In particular, although a proxy may coincide with its
historian, this is not the case if the historian is black.
The data structure is initialized by coloring $r$
white (i.e., by setting the root bits accordingly).

In terms of the abstract array $a$, the leaf colors white
and black signify ``not yet written to, and therefore still
containing the initial value 0'' and ``written to at least once'',
respectively.
For the actual array $A$, this translates approximately into
white and black meaning ``not initialized'' and
``initialized to a meaningful value'', respectively.

The data structure does not explicitly store the color
of any node except $r$.
Instead node colors must be deduced from histories.
It turns out that the colors of all nodes other than $r$ are implied by
the histories of the light paths.
A white leaf $\ell$ offers potential for storing a history
(namely in its associated word $A[\ell]$),
but we cannot know in advance where to find a white leaf.
This motivates the 
introduction of historians and proxies.
We actually need the history of a light path $P$ when,
during a descent in $T$ from $r$ to a leaf,
we reach the top node of $P$.
The historian of $P$ provides a fixed place (namely at the
leftmost leaf descendant) at which to look for the history,
but if the historian is black, then its own data must
be accommodated somewhere else---this is the role
of the (white) proxy.
How this works is perhaps best illustrated
by the following detailed description
of the realization of the operation \Tvn{read},
which basically carries out a descent in~$T$.
The call $\Tvn{leftmostleaf}(u)$ is assumed to
return (the integer identified with)
the leftmost leaf descendant of the node~$u$.

\goodbreak

\begin{tabbing}
\quad\=\quad\=\quad\=\quad\=\quad\=\quad\=\kill
$\Tvn{read}(\ell)$:\\
\>$u:=r$; $(*$ start at the root $*)$\\
\>\textbf{while} $u$ is gray \textbf{do}\\
\>\>\textbf{if} $u$ is a top node \textbf{then}
 $(*$ switch to a new history $*)$\\
\>\>\>$h:=\Tvn{leftmostleaf}(u)$; $(*$ $u$'s historian $*)$\\
\>\>\>$H:=A[h]$; $(*$ $u$'s history $*)$\\
\>\>$u:=\Tvn{viachild}(u,\ell)$; $(*$ continue towards $\ell$ $*)$\\
\>\textbf{if} $u$ is white \textbf{then return} 0;
 $(*$ the initial value $*)$\\
\>$(*$ now $\ell$ is black $*)$\\
\>\textbf{if} $u=r$ or $\ell\not=h$
 \textbf{then return} $A[\ell]$; $(*$ $\ell$ is neither a historian
 nor a proxy $*)$\\
\>$(*$ now $\ell$ is a black historian $*)$\\
\>\textbf{return} $A[p]$,
 where $p$ is the leaf at the end of the light path
that contains $u$'s parent;
\end{tabbing}

The procedure discovers a white ancestor of~$\ell$
and returns 0, determines that $\ell$ is black
and not a historian and returns $A[\ell]$,
or identifies $\ell$ as a black historian
and returns $A[p]$, where $p$ is the
proxy of $\ell$.
In all cases, the return value is correct.

Whenever the color of a node $u$ is queried, either $u=r$, in which
case the color of $u$ is given by the root bits, or
the color of $u$ can be deduced in constant time from
the history stored in $H$, one of whose elements is
the navigation vector of the parent of~$u$.
Similarly, if $u\not=r$, we can decide in constant time whether
$u$ is a top node by looking at the navigation vector of its parent.
The light path that contains $u$'s parent
can be followed in constant time per node, again
by inspection of~$H$.
Thus $\Tvn{read}$ can be executed
in $O(t)$ time.

To execute $\Tvn{write}(\ell,x)$, we carry out two phases.
The purpose of the first phase is to take the
data structure to a legal state in which $\ell$
is black and all values of the client sequence
$(x_0,\ldots,x_{n-1})$ except possibly $x_\ell$ are correct,
i.e., unchanged.
The second phase concludes the writing
by setting $x_\ell$ to $x$.
In the description of the two phases, we leave to
the reader details such as how to determine the
color of a given node;
in all cases, one can proceed similarly
as in the implementation of \Tvn{read}.

The first phase begins by following the path $P$ in $T$
from $r$ to $\ell$ until encountering a node
that is not gray.
This can be done similarly as in the implementation
of $\Tvn{read}$:
Each node visited is tested for
being a top node, and at each top node a new
history is fetched and subsequently used.
This computation, in particular, can determine the
color of $\ell$.
If $\ell$ is already black, the first phase terminates
without modifying the data structure.
Assume in the remaining discussion of the first phase
that $\ell$ is white and consider the consequences of
an \emph{update} that changes the color of $\ell$ from
white to black.
We will use the terms ``old'' and ``new'' to
describe the situation before and after the update,
respectively.

Because the color of an inner node in $T$ is a
function of the colors of its children, only
nodes on $P$ can change their color as a
result of the update.
The first phase proceeds to find the first node
$v$ on $P$
(i.e., the node on $P$ of minimal depth)
that changes its color.
The following observations show that this can
be done in a single traversal of~$P$
and characterizes the possible scenarios in
a useful way.
If some proper ancestor of $\ell$
is white (before the update), all proper ancestors of the
first white node $\widetilde{v}$ on $P$ are gray both before and
after the update, and all descendants of $\widetilde{v}$
on $P$ other than $\ell$ are white before and gray after the update.
Thus $v=\widetilde{v}$.
In the opposite case, namely if all proper ancestors
of $\ell$ are gray, let $\overline{v}$ be the last node
on $P$ that has a white or gray sibling if there is at
least one such node, and take $\overline{v}=r$ otherwise.
It is easy to see that all proper ancestors of
$\overline{v}$ are gray both before and after the
update and, by backwards induction on $P$, that
all descendants of $\overline{v}$,
including $\overline{v}$ itself, are black after the update.
In this case, therefore, $v=\overline{v}$.

As can be seen from the observations above, no descendant
of $v$ has more than one gray child before or after
the update under consideration.
Therefore the only node in $T_v$ that can be
a top node before or after the update is $v$ itself,
the only node in $T_v$ that can be a historian
before or after the update is the leftmost
leaf descendant $h_v$ of $v$, and before as
well as after the update at most one node
in $T_v$ is a proxy.
Moreover, at most one node in $T$ other than $v$ can become
or stop being a top node as a result of the update,
and this node, if it exists, must be the leftmost
gray sibling of $v$ and to the right of~$v$.

If $v=r$, change the root bits to reflect
the new color of the root.
Otherwise compute $u$ as the top node of the
light path that contains the (gray) parent of~$v$,
let $h_u$ be the historian of $u$ (before and
after the update) and let $p_u$ and $p'_u$
be the proxies of $u$ before and after the
update, respectively, which can be found by
following the old and new light paths that
start at $u$.
Store the new history of $p'_u$ in $A[h_u]$.
In particular, this registers the new color of~$v$.
To compute the history, it suffices to record
the new navigation vectors encountered on the
path in $T$ from $u$ to $p'_u$.
Now consider five cases 
that together cover all possible situations
and do not overlap. 
Even though every color change is irreversible,
Cases 1 and 2 show some aspects of being
reverses of each other, and so do
Cases 4 and~5. These four cases are illustrated in
Fig.~\ref{fig:setcolor}.

\begin{figure}[b!]
\begin{center}
\scalebox{1}[0.93]{%
\includegraphics{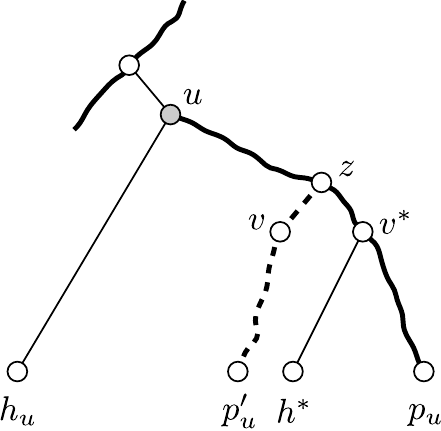}}\hspace{22mm}
\scalebox{1}[0.93]{%
\includegraphics{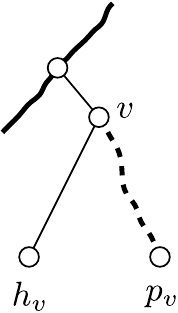}}
\end{center}
\caption{Left: The situation of Case 1 and, after a swap
of the labels ``$p_u$'' and ``$p'_u$'', also
of Case 2.
Right:
The situation of Cases 4 and 5.}
\label{fig:setcolor}
\end{figure}

\emph{Case~$1$:}
$v$ has a parent $z$ and is the preferred child of~$z$
after the update.
Thus $v$ changes its color from white to gray without
becoming a top node.
If $h_u$ is black before the update, then
execute $A[p'_u]:=A[p_u]$,
which moves $x_{h_u}$ from
the old to the new proxy of $h_u$.
This overwrites no relevant information, as
$p'_u$ is white and neither a historian nor a
proxy before the update unless
$p'_u$ coincides with $h_u$ or $p_u$,
in which case the assignment is not carried out
or has no effect.
Let $v^*$ be the preferred child of
$z$ before the update
and let $h^*$ be the leftmost
leaf descendant of $v^*$.
If $v^*$ is white before the update
(this includes the case $v^*=v$),
nothing more needs to be done.
If $v^*$ is gray (before and
after the update), it is a right sibling of $v$,
and it becomes a new top node
whose historian $h^*$
and proxy $p_u$ must have their associated
information updated accordingly.
To this end first execute
$A[p_u]:=A[h^*]$ and subsequently store
in $A[h^*]$ the history of the new
light path that starts at $v^*$ and
ends at $p_u$.
If $p_u=h^*$, the two assignments write
to the same word, but then any relevant
information present in $A[h^*]$ before the update
was already copied to $A[p'_u]$.

\emph{Case~$2$:}
$v$ has a parent $z$ and
is the preferred child of $z$ before,
but not after the update.
After the update, $v$ is black and no
descendant of $v$ is a historian or a proxy,
except that $h_v$
may coincide with $h_u$.
Let $v^*$ be the preferred child of
$z$ after the update
and let $h^*$ be its leftmost
leaf descendant.
If $v^*$ is gray, it is a right sibling
of $v$ and a top node
with historian $h^*$ and proxy
$p'_u$ before the update, whereas
after the update $p'_u$ is the proxy
of $u$ and $h^*$ is neither a historian
nor a proxy unless $h^*=p'_u$.
If
$h^*$ is black,
then execute $A[h^*]:=A[p'_u]$, which
moves $x_{h^*}$ to the correct place
and overwrites a history that is no longer useful.
Finally, independently of
the color of $v^*$ and
as in Case~1,
if $h_u$ is black,
then execute $A[p'_u]:=A[p_u]$.

In the remaining cases 3--5 $v$ is a preferred child
neither before nor after the update, so
there are no changes to light paths
outside of $T_v$
(i.e., the set of light edges outside
of $T_v$ remains the same).
In particular, $p'_u=p_u$.
Moreover, $v$ is not a leftmost child.

\emph{Case~$3$:}
$v$ is a leaf with at least one white left sibling.
There are no changes to light paths, so
nothing needs to be done.

\emph{Case~$4$:} $v$ is a top node after the update.
Before the update, $v$ is white, so no
descendant of $v$ is a historian or a proxy
at that time (informally, no information
is stored below~$v$).
Compute the proxy $p_v$
of $v$ after the update and store the
new history of $p_v$ in $A[h_v]$.
This involves following the new light path
that starts at $v$ and recording the new
navigation vectors encountered on the way.

\emph{Case~$5$:} $v$ is a top node before the update.
Because $v$ is black after the update, no descendant of $v$
is a historian or a proxy at that time.
Before the update, since $\ell$ is the only white
descendant of $v$, it is its proxy.
If $h_v$ is black (i.e., if $h_v\not=\ell$),
then copy the value of $A[\ell]$, namely $x_{h_v}$,
to $A[h_v]$.
This overwrites an old history that is no longer useful.

The second phase of the execution of
$\Tvn{write}(\ell,x)$ simulates the execution
of $\Tvn{read}(\ell)$ until the point when
the routine is ready to return as its answer
the value of $A[i]$ for some $i$
(that is either $\ell$ or the proxy of $\ell$).
Instead of returning $A[i]$, it finishes by
storing $x$ in $A[i]$.
Since $i$ is not a historian,
it is easy to see that a subsequent call of
$\Tvn{read}(\ell)$ will return $x$
and that the update of $A[i]$ leaves the data structure in
a legal state and does not change the value
of any elements of the client sequence
$(x_0,\ldots,x_{n-1})$ except $x_\ell$.
\end{proof}

The next lemma and its proof show how to handle
the case of an ``incomplete tree'' elegantly and,
following the initialization, without any
overhead to test for special cases.

\begin{lemma}
\label{lem:partialtree}
There is a clearable word array that,
for all given $n,d,t\in\mathbb{N}$
with $2 d t\le w$ and $n\le d^t$
such that $d$ is a power of\/~$2$, can be initialized
for universe size $n$ in constant time and subsequently
occupies $n w+2$ bits and,
if given access to $d$ and $t$,
supports \Tvn{read} and \Tvn{write} in $O(t)$ time.
\end{lemma}

\begin{proof}
We use the construction of the previous proof
for universe size $d^t$, but provide for its
storage only a word array $A$ with index set
$\{0,\ldots,n-1\}$ in addition to two root bits.
If $n=d^t$, nothing more needs to be said.
If $n<d^t$, before executing any true
\Tvn{write} operation, we change the color
of the root from white to gray
(of course, by modifying the root bits)
and store in $A[0]$ a history that corresponds
to the leaves $0,\ldots,n-1$ being white and
$n,\ldots,d^t-1$ being black.
Provided that only legal accesses are subsequently attempted,
this prevents the data structure from ever
choosing a proxy larger than $n-1$, and it
will process the operations correctly
without ever attempting to access one of
the nonexisting array elements
$A[n],\ldots,A[d^t-1]$.

The computational steps just described are conceptually
part of the initialization of the data structure,
but the computation of the history
to be stored in $A[0]$ may take more than constant time.
In order to guarantee a constant initialization time,
we postpone the steps and execute them
as an initial part of the first and only execution of
a \Tvn{write} operation that begins with a white root,
until which point we remember $n$ in~$A[0]$.
Since the steps are easily carried out in
$O(t)$ time, the bound of $O(t)$ for the
execution time of \Tvn{write} remains valid.
\end{proof}

We now take the step to values of~$n$
larger than $d^t$.

\begin{lemma}
\label{lem:manytrees}
There is a clearable word array that,
for all given $n,t\in\mathbb{N}$,
can be initialized for universe size~$n$
in constant time and
subsequently occupies at most $n w+\lceil{n({t/{(2 w)}})^t}\rceil$
bits and, if given access to $n$ and $t$,
supports \Tvn{read} and \Tvn{write}
in $O(t)$ time.
\end{lemma}

\begin{proof}
When $c\in\mathbb{N}$ is an arbitrary constant,
we can assume without loss of generality
that $n$ is a multiple of $c$.
This is because we can initialize up to
$c-1$ ``left-over'' words in constant time.
Moreover, a word RAM with a word length of $w$
bits can simulate one with a word length of
$c w$ bits with constant slowdown, i.e.,
every instruction can be simulated in constant time.
By these observations, we can essentially
pretend to be working on a word RAM with
a word length of $c w$ bits
(of course, the values communicated to and from
a user of the data structure are still $w$-bit
quantities).
In particular, we view $A$ as consisting
of $n/c$ \emph{large words} of $c w$ bits each,
and the condition $2 d t\le w$
of Lemma~\ref{lem:partialtree}
can be relaxed to $2 d t\le c w$.
We use this with $c=16$, for which choice the
condition becomes $d\le{{8 w}/t}$.

Assume that $t\le w$.
This entails no loss of generality because reducing
larger values of $t$ to $w$ does not increase
the space bound of the lemma
(recall that $w\ge\lceil{\log_2 n}\rceil$).
Compute $d$ as the largest power of~$2$
no larger than ${{8 w}/t}$ and
note that $d\ge{{4 w}/t}\ge 2$.
Dividing the universe $\{0,\ldots,{n/c}-1\}$
into \emph{ranges} of $d^t$ consecutive
elements each, except that the last range
may be smaller,
we store each subsequence of the client sequence
corresponding to a range in an instance of
the data structure of
Lemma~\ref{lem:partialtree}, called a \emph{tree},
except that the root bits are handled slightly differently.
Altogether we have
$N=\lceil{{n/{(c d^t)}}}\rceil\le\lceil{n({t/{(2 w)}})^t}\rceil$
trees.

If $N=1$, i.e., if there is only a single tree,
we use a single root bit to distinguish
between black and nonblack (i.e., white or gray).
In order to indicate a white root, in addition
to initializing the root bit to
the value that
denotes a nonblack color,
we store in $A[0]$
a value that cannot be the history of a
gray root, such as one in which all colors
in the navigation vector of the root are white.
The total redundancy is
$1=N\le\lceil{n({t/{(2 w)}})^t}\rceil$.

If $N>1$, we solve the problem of initializing
the $N$ trees differently.
Each tree has two root bits, and we must set these
to indicate a white root.
Assume, for convenience, that the root color white
is represented through two bits with a value of zero.
Then the task is to clear the $2 N$ root bits, i.e.,
to set them to zero.
Pack the $2 N$ root bits tightly in
$M=\lfloor{{{2 N}/w}}\rfloor$ fully occupied words
and at most one partially occupied word.
If there is an only partially occupied word,
clear it explicitly.
As for the $M$ fully occupied words,
maintain these, if $M\ge 1$, in a clearable word array
implemented with the folklore method discussed
near the end of Section~\ref{sec:related}.
In addition to the $M$ words, this needs space
for two tables with altogether $2 M$ entries and
one counter that takes values in $\{0,\ldots,M\}$.
Each table entry fits in a $w$-bit word, and
except in the trivial case
$w=1$, the counter can be stored in $N$ bits, so the
redundancy is at most $3 M w+N\le 7 N$.
Since $N>1$, we even have
$8 N\le 16 ({n/c})({t/{(4 w)}})^t
=n({t/{(4 w)}})^t
\le \lceil{n({t/{(2 w)}})^t}\rceil$.
This slightly stronger bound
is irrelevant here, but
useful
in the proof of Theorem~\ref{thm:main}.
\end{proof}

In order to derive 
Theorem~\ref{thm:main} from
Lemma~\ref{lem:manytrees} and
its proof, we must show how
to ``hide'' the parameters $n$ and $t$ in the
data structure essentially without additional space
or how to make do without them.
To achieve this, we pay close
attention to the layout of data within the
data structure.

\def\proofname{Proof of Theorem~\ref{thm:main}}
\begin{proof}
The memory allocated to the data structure
begins with a \emph{discriminator bit}
that selects between different representations.
If the discriminator bit is~0, it is simply followed by
a word array $A$ with index set $\{0,\ldots,n-1\}$
such that $A[\ell]=x_\ell$ for $\ell=0,\ldots,n-1$,
where $(x_0,\ldots,x_{n-1})$ is the
client sequence.
In other words, except for the discriminator
bit, the clearable word array is represented
as a usual array.
We call this representation
the \emph{all-black representation}.
The redundancy of the all-black representation
is~1, and it trivially supports \Tvn{read}
and \Tvn{write} in constant time
(but range violations cannot be detected).
Informally, the all-black representation is the
``terminal'' representation that can be used
when every array element has been written
to at least once.

If the discriminator bit is~1, we use slightly
modified variants of the data structure of
Lemma~\ref{lem:manytrees}.
Recall that the data structure operates with
large words of $c w$ bits each,
where $c=16$.
The first modification is to replace $c$
by $c'=c+6$, i.e., to redefine a large word
to consist of $c' w$ bits rather than $c w$ bits.
The parameter $d$ is computed
exactly as before, whereas $N$ is now
$\Tceil{{n/{(c' d^t)}}}$.
The modification does not affect the proof of
Lemma~\ref{lem:manytrees}, but it means that
a large word used to hold a history now has
$6 w$ bits of \emph{free space}
in which other information can be stored.
We follow the discriminator bit by
two words (i.e., groups of $w$ consecutive bits)
that hold $n$ and $t$.
Given these parameters, the quantities $d$ and
$N$ of the proof of
Lemma~\ref{lem:manytrees} can be computed
in constant time, so that they need not be stored
($N$ is computed according to the modified
formula above that involves $c'$ rather than~$c$).
Depending on the size of $N$, one of two
representations is chosen.

If $N<2 w$, we use the
\emph{few-roots representation}, in which
the words that hold $n$ and $t$ are followed
by $2 N$ \emph{root bits}, two for each of
the $N$ trees of
the modified data structure of
Lemma~\ref{lem:manytrees}.
Because the root bits are so few, they
can be cleared explicitly in constant time,
so there is no need to appeal to the
folklore method.
The final data component
of the few-roots representation
consists of $N$ \emph{segments}, each of
which corresponds to one of the $N$ trees.
If the universe size of a tree is $m$
(except in the case of the last tree,
$m=d^t$),
the corresponding
segment contains $m$ large words.

It is important to note that the first of
the $N$ segments begins immediately after
the discriminator bit.
Thus the two words that hold $n$ and $t$
as well as the at most four words that
hold the $2 N\le 4 w$ root bits are already
part of the first large word.
As observed above,
because of its six words of free space
that large word can still hold a history.
Of course, when the history changes, the
data stored in the free space
of the large word should not be touched.
The first large word is the only one
whose free space is actually used.

For brevity, say that the color of a tree
is the color of its root.
As long as a tree is not black,
its leftmost leaf is a historian,
so that the corresponding large word
indeed contains a history.
When the tree turns black, however,
all of its large words are
needed to hold elements of the client sequence.
This necessitates a final twist to the
few-roots representation:
It keeps two large words interchanged,
namely the first large word of the
first tree and the first large word of the
first nonblack tree.
In order to maintain this invariant, we
must be able to locate the correct
``replacement tree'' when what used
to be the first nonblack tree turns black.
Applying an algorithm of Lemma~\ref{lem:word}
to the sequence of root bits, this can be
done in constant time, as can the
``replacement'' itself, which is
a cyclic shift of two or three large words.
Of course, when the last tree turns black,
there is no ``replacement tree'', but then
we can return the two large words that
were kept interchanged to their original positions,
thereby allowing the information in the
free space to be overwritten, and be left
with the all-black representation.
Since $n$ is a multiple of $c'$ by assumption,
the redundancy
of the few-roots representation is~1.
The initialization time is constant, and
\Tvn{read} and \Tvn{write} are supported in
$O(t)$ time.

The final case to consider is when $N\ge 2 w$.
This is easy.
Recall from the end of the proof of
Lemma~\ref{lem:manytrees} that for $N>1$,
the data structure of that lemma is
smaller by at least $N$ bits than what is allowed
by the space bound of the lemma and of
Theorem~\ref{thm:main}.
Since the two words that hold $n$ and $t$
together occupy at most $N$ bits, we can simply follow
these words by a complete instance of
the data structure of
Lemma~\ref{lem:manytrees} and carry out
all operations in the latter,
of course obtaining the values of $n$ and $t$
from the two preceding words.
As was just argued, the redundancy of this
representation is bounded by
$\Tceil{n({t/{(2 w)}})^t}$,
the initialization time is constant,
and \Tvn{read} and \Tvn{write} are
supported in $O(t)$ time.
\end{proof}

It is interesting to note that we can add an additional
operation to our clearable word array, namely an
iteration that enumerates all first arguments
of past $\Tvn{write}$ operations
(informally, the positions to which writing
took place).
For this we would iterate over the
codes handed out by the folklore method
and the associated trees,
which is easy,
enumerate all leaves of each tree whose root
is black, and for each tree whose root is gray
carry out a depth-first search
(say) of its gray nodes and enumerate all leaf
descendants of their black children.
The time needed is proportional to the number $k$
of leaves enumerated plus the total number of
gray nodes, a quantity that is clearly bounded
by $(t+1)k$ and never larger than $2 n$.
The iteration must be called with an argument
that indicates~$n$.

\bibliography{clean}

\begin{thebibliography}{10}

\bibitem{AhoHU74}
Alfred~V. Aho, John~E. Hopcroft, and Jeffrey~D. Ullman.
\newblock {\em The Design and Analysis of Computer Algorithms}.
\newblock Addison-Wesley, 1974.

\bibitem{Ale10}
Andrei Alexandrescu.
\newblock {\em The {D} Programming Language}.
\newblock Addison-Wesley, 2010.

\bibitem{AngV79}
D.~Angluin and L.~G. Valiant.
\newblock Fast probabilistic algorithms for {H}amiltonian circuits and
  matchings.
\newblock {\em J. Comput. Syst. Sci.}, 18(2):155--193, 1979.
\newblock \href {http://dx.doi.org/10.1016/0022-0000(79)90045-X}
  {\path{doi:10.1016/0022-0000(79)90045-X}}.

\bibitem{FiaMNS91}
Amos Fiat, J.~Ian Munro, Moni Naor, Alejandro~A. Sch{\"{a}}ffer, Jeanette~P.
  Schmidt, and Alan Siegel.
\newblock An implicit data structure for searching a multikey table in
  logarithmic time.
\newblock {\em J. Comput. Syst. Sci.}, 43(3):406--424, 1991.
\newblock \href {http://dx.doi.org/10.1016/0022-0000(91)90022-W}
  {\path{doi:10.1016/0022-0000(91)90022-W}}.

\bibitem{FraG08}
Gianni Franceschini and Roberto Grossi.
\newblock No sorting? {B}etter searching!
\newblock {\em {ACM} Trans. Algorithms}, 4(1):2:1--2:13, 2008.
\newblock \href {http://dx.doi.org/10.1145/1328911.1328913}
  {\path{doi:10.1145/1328911.1328913}}.

\bibitem{FreW93}
Michael~L. Fredman and Dan~E. Willard.
\newblock Surpassing the information theoretic bound with fusion trees.
\newblock {\em J. Comput. Syst. Sci.}, 47(3):424--436, 1993.
\newblock \href {http://dx.doi.org/10.1016/0022-0000(93)90040-4}
  {\path{doi:10.1016/0022-0000(93)90040-4}}.

\bibitem{FreK16}
Kimmo Fredriksson and Pekka Kilpel{\"{a}}inen.
\newblock Practically efficient array initialization.
\newblock {\em J. Softw. Pract. Exper.}, 46(4):435--467, 2016.
\newblock \href {http://dx.doi.org/10.1002/spe.2314}
  {\path{doi:10.1002/spe.2314}}.

\bibitem{GosJSBB15}
James Gosling, Bill Joy, Guy Steele, Gilad Bracha, and Alex Buckley.
\newblock {\em The Java Language Specification, Java SE 8 Edition}.
\newblock Oracle America, 2015.

\bibitem{Hag98}
Torben Hagerup.
\newblock Sorting and searching on the word {RAM}.
\newblock In {\em Proc. 15th Annual Symposium on Theoretical Aspects of
  Computer Science ({STACS} 1998)}, volume 1373 of {\em LNCS}, pages 366--398.
  Springer, 1998.
\newblock \href {http://dx.doi.org/10.1007/BFb0028575}
  {\path{doi:10.1007/BFb0028575}}.

\bibitem{HagK16}
Torben Hagerup and Frank Kammer.
\newblock Succinct choice dictionaries.
\newblock {\em Computing Research Repository ({CoRR})}, arXiv:1604.06058
  [cs.DS], 2016.
\newblock \href {http://arxiv.org/abs/1604.06058} {\path{arXiv:1604.06058}}.

\bibitem{VHDL11}
{IEC/IEEE} {I}nternational {S}tandard; {B}ehavioural languages --- {P}art 1--1:
  {VHDL} {L}anguage {R}eference {M}anual.
\newblock IEC 61691--1--1:2011(E) {IEEE} Std 1076-2008, 2011.
\newblock \href {http://dx.doi.org/10.1109/IEEESTD.2011.5967868}
  {\path{doi:10.1109/IEEESTD.2011.5967868}}.

\bibitem{KatG17}
Takashi Katoh and Keisuke Goto.
\newblock In-place initializable arrays.
\newblock {\em Computing Research Repository ({CoRR})}, arXiv:1709.08900
  [cs.DS], 2017.
\newblock \href {http://arxiv.org/abs/1709.08900} {\path{arXiv:1709.08900}}.

\bibitem{LooNY17}
Jacob Teo~Por Loong, Jelani Nelson, and Huacheng Yu.
\newblock Fillable arrays with constant time operations and a single bit of
  redundancy.
\newblock {\em Computing Research Repository ({CoRR})}, arXiv:1709.09574
  [cs.DS], 2017.
\newblock \href {http://arxiv.org/abs/1709.09574} {\path{arXiv:1709.09574}}.

\bibitem{Mun86}
J.~Ian Munro.
\newblock An implicit data structure supporting insertion, deletion, and search
  in ${O}(\log^2 n)$ time.
\newblock {\em J. Comput. Syst. Sci.}, 33(1):66--74, 1986.
\newblock \href {http://dx.doi.org/10.1016/0022-0000(86)90043-7}
  {\path{doi:10.1016/0022-0000(86)90043-7}}.

\bibitem{Nav14}
Gonzalo Navarro.
\newblock Spaces, trees, and colors: The algorithmic landscape of document
  retrieval on sequences.
\newblock {\em {ACM} Comput. Surv.}, 46(4):52:1--52:47, 2014.
\newblock \href {http://dx.doi.org/10.1145/2535933}
  {\path{doi:10.1145/2535933}}.

\end{thebibliography}
\end{document}